\documentclass{lmcs} 
\pdfoutput=1

\usepackage[utf8]{inputenc}

\usepackage{lastpage}
\lmcsdoi{20}{3}{16}
\lmcsheading{}{\pageref{LastPage}}{}{}%
{Jul.~01,~2022}{Aug.~27,~2024}{}


\usepackage{hyperref}
\usepackage{amssymb}
\usepackage{algorithm}
\usepackage[noend]{algorithmic}
\usepackage{mymac}

\newcommand{\fN}{\mathfrak{N}}
\newcommand{\fNinf}{\mathfrak{N}_{\infty}}
\newcommand{\var}{\text{var}}

\begin{document}

\title[$\gftwo$ with local Presburger constraints]{On two-variable 
guarded fragment logic with expressive local Presburger constraints}

\author[C.-H.~Lu]{Chia-Hsuan Lu\lmcsorcid{0009-0008-5446-9150}}[a]
\author[T.~Tan]{Tony Tan\lmcsorcid{0009-0005-8341-2004}}[b]

\address{Department of Computer Science, University of Oxford, 
United Kingdom.}	
\email{chia-hsuan.lu@cs.ox.ac.uk}  

\address{Department of Computer Science, University of Liverpool, 
United Kingdom.}	
\email{tonytan@liverpool.ac.uk}  





\begin{abstract}
We consider the extension of the two-variable guarded fragment logic
with local Presburger quantifiers.
These are quantifiers that can express properties such as 
``the number of incoming blue edges plus 
twice the number of outgoing red edges
is at most three times the number of incoming green edges''
and captures various description logics with counting,
but without constant symbols.
We show that the satisfiability problem for this logic 
is $\exp$-complete.
While the lower bound already holds for 
the standard two-variable guarded fragment logic,
the upper bound is established by a novel, 
yet simple deterministic graph-based algorithm.
\end{abstract}

\maketitle

\section{Introduction}
\label{sec:intro}

In this paper we consider the extension of two-variable guarded fragment logic
with the so-called {\em local Presburger quantifiers}, 
which we denote by $\gptwo$.
These are quantifiers that can express local numerical properties such as
``the number of outgoing red edges plus 
twice the number of incoming green edges
is at most three times the number of outgoing blue edges.''
It was first considered by Bednarczyk, et. al.~\cite{BOPT21} and
trivially subsumes $\gctwo$ 
(two-variable guarded fragments with counting quantifiers),
which captures various description logics with counting such as
$\mathcal{ALCIHQ}$ and
$\mathcal{ALCIH}b^{\textsf{self}}$~\cite{dl-textbook,dl-handbook03,gradel-dl}.
Bednarczyk, et. al.~\cite{BOPT21} showed that 
both satisfiability and finite satisfiability problems are in 3-$\nexp$
by reduction to the two-variable logic with counting quantifiers.

We show that the satisfiability problem for this logic 
is $\exp$-complete.
The lower bound is already known for the standard 
two-variable guarded fragments~[Corollary~4.6]\cite{Gradel99-guarded}.
Our contribution is the upper bound,
which is established by a novel, yet simple 
deterministic exponential time algorithm which works similarly 
to the type elimination approach, 
first introduced by Pratt~\cite{Pratt79}.
Intuitively, it starts by representing the input sentence
as a graph whose vertices and edges represent the allowed types.
It then successively eliminates the vertex or edge
that contradicts the input sentence until 
there is no more vertex or edge to eliminate.

Our algorithm has a markedly different flavour 
from the standard tableaux method usually used to establish
the upper bound of guarded fragments.
Note that the tableaux method works by exploiting the so-called 
{\em tree-like model} property,
where it tries to construct a tree-like model using 
a polynomial space alternating Turing machine and
a configuration of the Turing machine corresponds to
an element in the model.
To apply this method, it is essential that 
there is only a polynomial bound on the branching degree of 
each node in the tree
and it is not clear whether this bound still holds for $\gptwo$.
In fact, already for $\gctwo$,
the degree can be exponential 
(when the counting quantifiers are encoded in binary).
To circumvent the exponential blow-up,
the $\exp$ upper bound for the satisfiability problem of $\gctwo$
is obtained by first
reducing it to three-variable guarded fragments,
before applying the tableaux method~\cite{Kazakov-gf}.
It was not clear {\em a priori} 
how this technique can be extended to the satisfiability of~$\gptwo$.


\paragraph*{Acknowledgement.}
Recently Bednarczyk and Fiuk~\cite{BF22} independently
obtained the same $\exp$ upper bound for the satisfiability 
problem for $\gptwo$.
Their proof uses the tableaux method.
To avoid the exponential branching degree, 
they restrict each node in the tableau to correspond only 
to the \emph{type} of an element in a model and 
it branches only to the different types of the children element. 
This method yields an alternating polynomial space algorithm,
hence, a deterministic exponential time algorithm.
A brief comparison between their algorithm and ours 
is presented in Section~\ref{subsec:comparison}.

\paragraph*{Other related work.}
The guarded fragment is one of the most prominent decidable 
fragments of first-order logic~\cite{AndrekaNB98}.
The satisfiability problem is 2-$\exp$-complete and 
it becomes $\exp$-complete
when the number of variables or the arity of the signature is 
fixed~[Theorem~4.4 and Corollary~4.6]\cite{Gradel99-guarded}.
Various description (DL) and modal (ML) logics 
are captured by the fragment when 
the arity is fixed to two~\cite{gradel-dl,dl-handbook03,dl-textbook}.
The key reason for the decidability of the guarded fragment 
is the {\em tree-like model property}
which allows the application of the tableaux method~\cite{Vardi96}.

Recently, a deterministic exponential time algorithm for 
a fragment of the two-variable logic was proposed
by Lin, et. al.~\cite{llt21}.
The algorithm there is also a graph-based algorithm,
and has a similar flavour to type-elimination algorithm.
However, it is not clear, {\em a priori}, 
how Lin, et. al.'s algorithm can be extended to $\gptwo$.

Pratt-Hartmann~\cite{Pratt-Hartmann-gf} proposed 
an elegant reduction of the satisfiability problem for $\gctwo$ formulas
to the solvability of (exponential size) homogeneous instances of 
Integer Linear Programming (ILP)
which are of the form $A\vx=0 \wedge B\vx\geq \vc$, 
where $A$ and $B$ are matrices with integer entries,
$\vx$ is a (column) vector of variables and 
$\vc$ is a vector of integers.
To check whether $A\vx=0 \wedge B\vx\geq \vc$ 
admits a solution in $\bbN$,
it is sufficient to check whether it admits a solution in $\bbQ$,
which is known to be in $\ptime$.
This implies the $\exp$ upper bound for both 
the satisfiability and finite satisfiability problems for $\gctwo$.

In general, $\gptwo$ captures various description logics
with counting such as $\mathcal{ALCIHQ}$ and 
$\mathcal{ALCIH}b^{\textsf{self}}$,
but without nominals~\cite{dl-handbook03}.
Note that allowing nominals in $\gctwo$ makes the complexity
of the satisfiability problem becomes $\nexp$-complete~\cite{Tobies01}.

Some logics that allow similar quantifiers as 
the local Presburger quantifiers
were proposed and studied by various 
researchers~\cite{baader17,bbr20,demri10,kupke10}.
The decidability result is obtained by the tableaux method,
but their logics do not allow the inverses of binary relations.

The extension of {\em one-variable logic} with quantifiers of the form
$\exists^S x  \ \phi(x)$, where $S$ is a ultimately periodic set, 
is $\np$-complete~\cite{bartosztcs}.
This logic is similar to the quantifier-free fragment of 
Boolean Algebra with Presburger Arithmetic (QFBAPA) introduced by
Kuncak and Rinard~\cite{bapa07}.
Semantically $\exists^S x  \ \phi(x)$ means 
the number of $x$ where $\phi(x)$ holds is in the set $S$.
The extension of two-variable logic with such quantifiers 
is later shown to be 
$2$-$\nexp$ by Benedikt, et. al.~\cite{BKT20}
and the proof makes heavy use of the {\em biregular graph method} 
introduced by Kopczynski and Tan~\cite{KT15}
to analyse the spectrum of two-variable logic with counting quantifiers.
These proofs and results do not apply in our setting since
the logics they considered already subsume the two-variable logic.

%
%

\paragraph*{Organisation}
This paper is organised as follows.
In Section~\ref{sec:presburger} we define \emph{linear constraints} and 
review two structures of natural numbers 
for their semantics.
One structure is the standard structure of natural numbers and 
the other is its extension with $\infty$ value.
In Section~\ref{sec:gptwo} we present the formal definition of $\gptwo$. 
The main result is presented in Section~\ref{sec:main}.
We conclude with Section~\ref{sec:conclusion}.

\section{Presburger arithmetic}
\label{sec:presburger}

Let $\bbN$ denote the set of natural numbers including $0$.

%
%
%

\paragraph*{Linear constraints}
We assume a countable infinite set $\mathcal{X}$ of variables.
A \emph{linear constraint} $\xi$ is a constraint of the form:
\begin{align*}
\label{eq:lin-cons}
\xi & \quad := \quad
\kappa_1 x_1 \ + \ \cdots \ + \ \kappa_k x_k \quad \circledast \quad 
\delta \ + \ 
\lambda_1 y_1 \ + \ \cdots \ + \ \lambda_\ell y_\ell,
\end{align*}
where $x_1,\ldots,x_k,y_1,\ldots,y_{\ell}$ 
are variables from $\mathcal{X}$,
all $\delta,\kappa_1,\ldots,\kappa_k,\lambda_1,\ldots,\lambda_{\ell}$ 
are from $\bbN$
and $\circledast$ is one of the symbols 
$=$, $\neq$, $\leq$, $\geq$, $<$, $>$.
$\equiv_d$ or $\not\equiv_d$, 
where $d \in \bbN-\{0\}$.
The relation $\equiv_d$ is intended to mean congruence modulo $d$, 
i .e., $n\equiv_d m$ if and only if
$m,n\in \bbN$ and there is $k_1,k_2\in \bbN$ such that $n+k_1d=m+k_2d$.

We let $\var(\xi)$ denote the set of variables that 
appear in the constraint $\xi$ mentioned above, i.e.,
$\var(\xi)=\{x_1,\ldots,x_k,y_1,\ldots,y_{\ell}\}$.
The coefficients in $\xi$ are the natural numbers 
$\delta,\kappa_1,\ldots,\kappa_k,\lambda_1,\ldots,\lambda_{\ell}$
which include $d$ when $\circledast$ is $\equiv_d$ or $\not\equiv_d$.

A finite set $\cC$ of linear constraints is also called 
a \emph{system of linear constraints},
or in short, {\em system}.
When convenient, we will also view a system $\cC$ 
as a conjunction of its constituent linear constraints.
We let $\var(\cC)$ denote the set $\bigcup_{\xi \in \cC} \var(\xi)$.
The coefficients in $\cC$ are the coefficients 
in the linear constraints in $\cC$.

We will consider two structures $\fN$ and $\fNinf$ in which 
the satisfaction of a linear constraint is defined.
Intuitively $\fN$ is the standard structure of natural numbers and 
$\fNinf$ is its extension with the value $\infty$.

\paragraph*{The structure $\fN$.}
The semantics of linear constraints can be naturally defined over 
the structure of natural numbers
$\fN = (\bbN, +, \cdot, \leq, 0, 1)$ where 
$+,\cdot, \leq, 0, 1$ are interpreted in the standard way.
Note that by using the relation $\leq$ and the equality predicate $=$,
we can define the relations $<$ and $\equiv_d$ in $\fN$ as follows.
\begin{itemize}
\item
$m < n$ if and only if $m\leq n$ and $m\neq n$, for every $n,m \in \bbN$.
\item
$n\equiv_d m$ if and only if 
there is $k_1,k_2\in \bbN$ such that $n+k_1d=m+k_2d$.
\end{itemize}

An {\em assignment} to the variables in $\xi$ is 
a mapping $F:\var(\xi)\to\bbN$.
It satisfies the linear constraint $\xi$, if $\xi$ holds in $\fN$ when 
each variable $x$ is assigned with the value $F(x)$.
A \emph{solution} to a system $\cC$ is
an assignment $F:\var(\cC)\to\bbN$ that satisfies 
every linear constraint in $\cC$.
If such a solution exists, 
we say that $\cC$ \emph{admits a solution} in $\fN$.
The following theorem establishes useful bounds on the solution
of a system of linear constraints.

\begin{thm}
\label{theo:caratheodory}
There are constants $c_1,c_2\in \bbN$ such that
for every system $\cC$ of linear constraints, the following holds
where $t=|\cC|$ and $M$ is the maximal coefficient 
in $\cC$.\footnote{The results by 
Eisenbrand and Shmonin~\cite[Theorem~1]{caratheodory-integer}
and Papadimitriou~\cite[Theorem]{papa-ilp} are stated 
in terms of Integer Linear Programming (ILP), 
which is a conjunction of a finite number of linear constraints 
not involving $\equiv_d$ and its negation $\not\equiv_d$. 
Since $a\equiv_d b$ is equivalent to 
$a+k_1d=b+k_2d)$ for some $k_1,k_2\in \bbN$,
it is obvious that the results also hold when $\equiv_d$ is involved.
Note that when $\cC$ contains a linear constraint involving $\equiv_d$,
the integer $d$ is considered a coefficient in $\cC$.}
\begin{enumerate}
\item\cite[Theorem~1]{caratheodory-integer}
If $\cC$ admits a solution in $\fN$,
then it admits a solution in $\fN$ in which
the number of variables assigned with non-zero values is at most 
$c_1t\log_2(c_2tM)$.
\item\cite[Theorem]{papa-ilp}\footnote{There is only one theorem 
in the paper by 
Papadimitriou~\cite{papa-ilp} and it is without a number.}
If $\cC$ admits a solution in $\fN$,
then it admits a solution in $\fN$ in which
every variable is assigned with a value at most $c_1t(tM)^{c_2t}$.
\end{enumerate}
\end{thm}

Combining (1) and (2) in Theorem~\ref{theo:caratheodory},
we obtain the following corollary,
which will be useful to establish the upper bound of our algorithm.

\begin{cor}
\label{cor:caratheodory}
There are constants $c_1,c_2\in \bbN$ such that
for every system $\cC$ of linear constraints, 
if $\cC$ admits a solution in $\fN$,
then it admits a solution in $\fN$ in which
the number of variables assigned with non-zero values is 
at most $c_1t\log_2(c_2tM)$
and every variable is assigned with a value at most $c_1t(tM)^{c_2t}$,
where $t=|\cC|$ and $M$ is the maximal coefficient in the system $\cC$.
\end{cor}
\begin{proof}
Let $c_1,c_2$ be the constant in Theorem~\ref{theo:caratheodory}.
Let $\cC$ be a system of linear constraints, where
$t=|\cC|$ and $M$ is the maximal coefficient in $\cC$.
Suppose $\cC$ admits a solution in $\fN$.
By~(1) in Theorem~\ref{theo:caratheodory},
it admits a solution in $\fN$ in which
the number of variables assigned with non-zero values is at most 
$c_1t\log_2(c_2tM)$.
Let $\mathcal{X}$ be the set of variables in $\var(\cC)$
that are assigned with non-zero values.
Thus, $|\mathcal{X}|\leq c_1t\log_2(c_2tM)$.

Let $\cC'$ be the system obtained from $\cC$ by 
assigning all the variables not in $\mathcal{X}$ with the zero value.
Thus, $\var(\cC')=\mathcal{X}$.
Let $t'=|\cC'|$ and $M'$ be the maximal coefficient in $\cC'$.
Note that $t'\leq t$ and $M'\leq M$
and that $\cC'$ admits a solution in $\fN$.
By~(2) in Theorem~\ref{theo:caratheodory},
$\cC'$ admits a solution in $\fN$ in which
every variable is assigned with a value at most 
$c_1t'(t'M')^{c_2t'}\leq c_1t(tM)^{c_2t}$.
This solution of $\cC'$ can be extended to a solution of $\cC$
where every variable not in $\mathcal{X}$ is assigned with zero value.
\end{proof}

\paragraph*{The structure $\fNinf$.}
We assume a constant symbol $\infty$ which 
symbolises the infinite value.
The extension of $\fN$ with $\infty$ is the structure 
$\fNinf = (\bbN \cup\{\infty\}, +,\cdot, \leq, 0, 1)$
where $+,\cdot,\leq$ involving elements in $\bbN$ 
are defined as in $\fN$.
When $\infty$ is involved, they are defined as follows.
$n < \infty$ and $n+\infty=\infty+n=\infty+\infty=\infty$ 
for every $n\in \bbN$,
$n\cdot\infty=\infty\cdot n = \infty$ for every $n\neq 0$,
and $0\cdot \infty = \infty\cdot 0=0$.
An \emph{assignment} is a mapping $F:\var(\xi)\to\bbN\cup\{\infty\}$.
The notion of a system admitting a solution in $\fNinf$ 
is defined in the similar manner as in $\fN$.

Note that the constraint $x=x+1$ does not admit a solution in $\fN$,
but admits a unique solution in $\fNinf$ 
where $x$ is assigned with the value $\infty$.
Thus, we can view the satisfiability of 
a system of linear constraints in $\fN$ 
as a special case of the satisfiability in $\fNinf$,
where we have the linear constraint $x\neq x+1$ for every variable $x$.

It is a folklore result that Corollary~\ref{cor:caratheodory} 
can be extended to $\fNinf$,
see, e.g.,~\cite[Section~2.3]{c2e-Pratt-Hartmann}.
We state it formally in Corollary~\ref{cor:caratheodory-inf}.


\begin{cor}
\label{cor:caratheodory-inf}
There are constants $c_1,c_2\in \bbN$ such that
for every system $\cC$ of linear constraints, 
if $\cC$ admits a solution in $\fNinf$,
then it admits a solution in $\fNinf$ in which
the number of variables assigned with non-zero values 
is at most $c_1t\log_2(c_2tM))$
and every variable is assigned with either $\infty$ or 
a value at most $c_1t(tM)^{c_2t}$,
where $t=|\cC|$ and $M$ is the maximal coefficient in $\cC$.
\end{cor}

\section{Two-variable guarded fragment with local Presburger 
quantifiers \texorpdfstring{$(\gptwo)$}{GP²}}
\label{sec:gptwo}

We fix a vocabulary $\Sigma$ consisting of 
only unary and binary predicates.
We do not allow for constant symbols and 
we will consider the logic with the equality predicate.
As usual, for a vector $\vx$ of variables,
$\varphi(\vx)$ denotes a formula $\varphi$ whose 
free variables are exactly those in $\vx$.

Let $\cA$ be a structure and let $a$ be an element in $\cA$.
For a formula $\varphi(x,y)$, we denote by $|\varphi(x,y)|^{x/a}_{\cA}$
the number of elements $b$ such that $\cA,x/a,y/b \models \varphi(x,y)$.
When the value $|\varphi(x,y)|^{x/a}_{\cA}$ 
is not finite, 
we write $|\varphi(x,y)|^{x/a}_{\cA}=\infty$.

\paragraph*{Local Presburger (LP) quantifiers.}
The \emph{local Presburger} (LP) quantifiers are quantifiers of the form:
\[
\cP(x) \quad :=\quad
\sum_{i=1}^n\ \lambda_i\cdot \#_y^{r_i}[\varphi_i(x,y)] 
\quad \circledast \quad 
\delta
\ + \ \sum_{i=1}^m\ \kappa_i\cdot \#_y^{s_i}[\psi_i(x,y)],
\]
where $\delta,\lambda_i,\kappa_i \in \bbN$;
each $r_i,s_i$ is an atom $R(x,y)$ or $R(y,x)$ 
for some binary relation $R$;
each $\varphi_i(x,y),\psi_i(x,y)$ is 
a formula with free variables $x$ and $y$;
and $\circledast$ is one of the symbols 
$=$, $\neq$, $\leq$, $\geq$, $<$, $>$, $\equiv_d$ or $\not\equiv_d$, 
where $d \in \bbN-\{0\}$.
Note that $\cP(x)$ has precisely one free variable $x$.
We say that {\em $\cP(x)$ holds in $\cA,x/a$}, 
denoted $\cA, x/a \models \cP(x)$, if 
the (in)equality $\circledast$ holds in $\fNinf$ when 
each $\#_y^{r_i}[\varphi_i(x,y)]$ and $\#_y^{s_i}[\psi_i(x,y)]$ 
are substituted with 
$|r_i(x,y)\wedge \varphi_i(x,y)|^{x/a}_{\cA}$ and
$|s_i(x,y)\wedge \psi_i(x,y)|^{x/a}_{\cA}$, respectively.

Note that the negation of the LP quantifier $\cP(x)$
stated above is obtained by changing the relation symbol $\circledast$
to its negated counterpart, i.e.,
the symbol $=$ is changed to $\neq$, 
$\leq$ to $>$, $\equiv_d$ to $\not\equiv_d$ and so on.
For example, the negation of \
$3\cdot \#_y^{R(x,y)}[\varphi(x,y)] \equiv_5 7$
\ is\  
$3\cdot \#_y^{R(x,y)}[\varphi(x,y)] \not\equiv_5 7$.

In the following,
to avoid clutter, we will allow the constants in an LP quantifier 
to be negative integers and write it in the form:
\[
\cP(x) \quad :=\quad
\sum_{i=1}^n\ \lambda_i\cdot \#_y^{r_i}[\varphi_i(x,y)] 
\quad \circledast \quad 
\delta
\]
When we do so, it is implicit that 
we mean the LP quantifier obtained after rearranging the terms
so that none of the coefficients are negative integers.
For example, when we write:
\[
3\cdot \#_y^{R(x,y)}[\varphi(x,y)]\ -\ 2\cdot \#_y^{S(y,x)}[\psi(x,y)]
\  =\ -7,
\]
we mean:
\[
2\cdot \#_y^{S(y,x)}[\psi(x,y)] \ = \
7 \ + \
3\cdot \#_y^{R(x,y)}[\varphi(x,y)].
\]

We say that a quantifier $\cP(x)$ is in {\em basic form}, 
if it is of the form:
\[
\cP(x) \quad :=\quad
\sum_{i=1}^n\ \lambda_i\cdot \#_y^{R_i(x,y)}[\varphi_i(x,y)] 
\quad \circledast \quad 
\delta,
\]
where each $\varphi_i(x,y)$ is either 
the equality $x=y$ or the inequality $x\neq y$.

\begin{rem}
With LP quantifiers one can state whether 
an element has finite or infinitely many outgoing edges.
Consider the following LP quantifier:
\begin{align*}
\cP_{\infty}(x) & \quad := \quad 
\#_y^{R(x,y)}[\top] = \#_y^{R(x,y)}[\top] +1.
\end{align*}
Since the equality only holds when $|R(x,y)|^{x/a}_{\cA}=\infty$,
it follows that 
$\cA,x/a \models \cP_{\infty}(x)$ if and only if 
there are infinitely many outgoing $R$-edges 
from $a$ in the structure $\cA$.

Note that the negation of $\cP_{\infty}(x)$ is:
\begin{align*}
\cP_{\text{fin}}(x) & \quad := \quad 
\#_y^{R(x,y)}[\top] \neq \#_y^{R(x,y)}[\top] +1.
\end{align*}
Thus, $\cA,x/a \models \cP_{\text{fin}}(x)$ if and only if 
there are finitely many outgoing $R$-edges from $a$ 
in the structure $\cA$.
\end{rem}



\paragraph*{The guarded fragment class.}
The class $\gf$ of guarded fragment logic is 
the smallest set of first-order formulas 
such that:\footnote{In this paper $\vx,\vy,\vz$ denote
vectors of variables and 
$\vx\subseteq \vz$ means that all the variables that 
appear in $\vx$ also appear in $\vz$.}
\begin{itemize}
\item
$\gf$ contains all atomic formulas $R(\vx)$ and 
equalities between variables.
\item
$\gf$ is closed under boolean combinations.
\item
If $\varphi(\vx)$ is in $\gf$, $R(\vz)$ is an atom
and $\vx,\vy\subseteq \vz$,
then both $\exists \vy\ R(\vz) \wedge \varphi(\vx)$ and 
$\forall \vy\ R(\vz) \to \varphi(\vx)$ are also in $\gf$.
\end{itemize}
We define the class $\gp$ to be the extension of $\gf$ 
with LP quantifiers,
i.e., by adding the following rule.
\begin{itemize}
\item
An LP quantifier
$\cP(x):=\ \sum_{i=1}^n\ \lambda_i\cdot \#_y^{r_i}[\varphi_i(x,y)] 
\ \circledast \ \delta$ is in $\gp$
if and only if each $\varphi_i(x,y)$ is in $\gp$.
\end{itemize}
We denote by $\gptwo$ the restriction of $\gp$ to 
formulas using only two variables: $x$ and $y$.
As before, when considering $\gptwo$ formulas, we assume that
the vocabulary is $\Sigma$, i.e., 
the arities of the predicates is at most $2$ and
there is no constant symbol.

\begin{rem}
\label{rem:local-quantifier}
The standard quantifiers $\forall$ and $\exists$
in $\gftwo$ can be expressed using LP quantifiers.
The universal quantifier $\forall y\ r(x,y) \to \varphi(x,y)$
is equivalent to:
$$
\#_y^{r(x,y)}[\neg\varphi(x,y)] \ =\ 0,
$$
and $\exists y\ r(x,y) \wedge \varphi(x,y)$
is equivalent to: 
$$
\#_y^{r(x,y)}[\varphi(x,y)]\ \geq\ 1.
$$
It is also for this reason that we call LP quantifiers 
``local Presburger quantifiers''
as it allows us the Presburger arithmetic reasoning 
on the neighbourhood of an element.
The term ``quantifier'' in this case is similar in spirit with 
H\"artig quantifiers (cardinality comparison quantifiers)~\cite{hartig,hartig-survey} and
ultimately periodic counting quantifiers~\cite{BKT20}.
\end{rem}

We denote by $\satgptwo$ the problem that
on input $\gptwo$ formula $\varphi$,
decides if it is satisfiable.
The following lemma states that
it suffices to consider only $\gptwo$ formulas in a normal form.

\begin{lem}
\label{lem:normal-form}
There is a linear time algorithm that converts a $\gptwo$ formula into 
an equisatisfiable first-order formula in the normal form 
(over an extended signature):
\begin{equation}
\label{eq:normal-form}
\Psi := \forall x\  \gamma(x) \quad \wedge\quad
\bigwedge_{i=1}^k \forall x \forall y\ \alpha_i(x,y)
\quad\wedge\quad
\bigwedge_{i=1}^\ell \forall x\  \big(  q_i(x)  \to  \cP_i(x)\big),
\end{equation}
where
\begin{itemize}
\item
$\gamma(x)$ is a quantifier-free formula,
\item
each $\alpha_i(x,y)$ is a quantifier-free formula of the form:
\[
(R(x,y)  \wedge  x\neq y) \ \to \ \beta(x,y)
\]
where $R(x,y)$ is an atomic formula and 
$\beta(x,y)$ is a quantifier-free formula,
\item
each $q_i(x)$ is an atomic formula,
\item
each $\cP_i(x)$ is an LP quantifier in the basic form.
\end{itemize} 
\end{lem}
\begin{proof} 
The proof uses a routine renaming technique 
(see, e.g.,~\cite{Kazakov-gf,Pratt-Hartmann-gf}).
For completeness, we present it here.
By pushing the negation inward, we can assume that 
the input formula is in negation normal form,
i.e., every negation is applied only on atomic formulas.
Intuitively we rename each subformula of one free variable
with a fresh unary predicate.
The renaming is done bottom-up starting from the subformula
with the lowest quantifier rank.
We consider four cases.
\begin{itemize}
\item
Case 1:
A subformula in the form of an LP quantifier:
\begin{align}
\label{eq:rename-case-1}
\cP(x) \quad & := \quad
\sum_{i=1}^n\ \lambda_i\cdot \#_y^{r_i}[\varphi_i(x,y)] 
\quad \circledast \quad 
\delta
\end{align}
where each $\varphi_i(x,y)$ is quantifier-free.
\item 
Case 2: A subformula with one free variable $x$ of the form:
\begin{align}
\label{eq:rename-case-3}
\psi(x)\quad  & := \quad \varphi_1(x)\ \land\ \varphi_2(x)
\end{align}
where $\varphi_1(x)$ and $\varphi_2(x)$ are quantifier-free.
\item 
Case 3: A subformula with one free variable $x$ of the form:
\begin{align}
\label{eq:rename-case-4}
\psi(x)\quad  & := \quad \forall y\ r(x,y) \to \varphi(x,y)
\end{align}
where $\varphi(x,y)$ is quantifier-free.
\item 
Case 4: A subformula with one free variable $x$ of the form:
\begin{align}
\label{eq:rename-case-5}
\psi(x) \quad & := \quad \exists y \ r(x,y) \ \wedge \ \varphi(x,y)
\end{align}
where $\varphi(x,y)$ is quantifier-free.
\end{itemize}
All the other cases can be handled in a similar manner.

We start with Case~1.
Let $\cP(x)$ be an LP quantifier in form of~(\ref{eq:rename-case-1}).
We will first convert it to the basic form.
We introduce new binary predicates $R_1,\ldots,R_n$
and rewrite $\cP(x)$ into $\cP'(x)$ as:
\begin{align*}
\cP'(x) \quad & := \quad
\sum_{i=1}^n\ \lambda_i\cdot \#_y^{R_i(x,y)}[x=y]\ +\ 
\lambda_i\cdot \#_y^{R_i(x,y)}[x\neq y] 
\ \circledast \ \delta.
\end{align*}
It is useful to note that $\cP'(x)$ is actually equivalent to 
$\sum_{i=1}^n \lambda_i\cdot \#_y^{R_i(x,y)}[\top]\ \circledast\ \delta$.

Obviously, $\cP'(x)$ is in the basic form.
Then we replace $\cP'(x)$ with a fresh unary predicate symbol $q(x)$ and 
add the following conjunct to the original formula.
$$
\forall x \ q(x)\ \to \ \cP'(x).
$$
We further add the conjunct that asserts that 
each $R_i(x,y)$ is equivalent to $r_i(x,y)\wedge \varphi_i(x,y)$:
\begin{align*}
& \bigwedge_{i=1}^n 
\forall x \forall y \ R_i(x,y)  \to  
\big(r_i(x,y)\wedge \varphi_i(x,y)\big)
\
\wedge\ 
\bigwedge_{i=1}^n
\forall x \forall y \ r_i(x,y)  \to  (\varphi_i(x,y)\ \to \ R_i(x,y)).
\end{align*}
Finally note that each sentence of the form 
$\forall x \forall y\ S(x,y)\to \zeta(x,y)$, 
where $S(x,y)$ is atomic and $\zeta(x,y)$ is quantifier-free, 
can be rewritten into the form:
\[
\forall x \ S(x,x)\to  \zeta(x,x)
\quad\wedge\quad
\forall x \forall y \ S(x,y)\wedge x\neq y  \to  \zeta(x,y).
\]
This finishes the renaming for Case~1.

For Case~2, i.e., 
a subformula $\psi(x)$ in the form of~(\ref{eq:rename-case-3}),
we replace $\psi(x)$ with a fresh unary predicate symbol $q(x)$ and
add the following conjunct to the original formula.
$$
\forall x \ q(x) \ \to \ \big(\varphi_1(x)\ \land\ \varphi_2(x)\big).
$$
Note that $\forall x\ q(x)\to (\varphi_1(x)\wedge\varphi_2(x))$
is in the form of the first part of (3.1).
Since $\bigwedge_i\forall x \gamma_i(x)$
is equivalent to $\forall x \bigwedge_i\gamma_i(x)$,
adding such conjunct does not violate the form of (3.1).

Finally, for Cases~3 and~4, i.e., 
a subformula $\psi(x)$ in the form of~(\ref{eq:rename-case-4})
or~(\ref{eq:rename-case-5}),
we can rewrite it into an LP quantifier
as in Remark~\ref{rem:local-quantifier}
and
proceed according to Case~1.

We perform such renaming procedure repeatedly until
we obtain a $\gptwo$ sentence in normal form~(\ref{eq:normal-form}).
This completes the proof of Lemma~\ref{lem:normal-form}.
\end{proof}

\begin{rem}
We also note that if the sentence $\Psi$ in 
normal form~(\ref{eq:normal-form}) is satisfiable,
then it is satisfied by an infinite model.
Indeed, let $\cA$ be a model of $\Psi$.
We make infinitely many copies of $\cA$, denoted by 
$\cA_1,\cA_2,\ldots$,
and disjoint union them all to obtain a new model $\cB$.
For every pair $(a,b)$, where $a$ and $b$ do not come from 
the same $\cA_i$,
we set $(a,b)$ not to be in any binary relation $R^{\cB}$.
It is routine to verify that $\cB$ still satisfies $\Psi$.

We also note that there is a $\gptwo$ sentence that
is satisfied only by infinite models.
Consider the following sentence
with one unary predicate $U$ and one binary predicate $R$.
\begin{align*}
\Phi \ & := \ \forall x\  U(x) \ \wedge\
\forall x \big(U(x) \ \to \ \#^{R(x,y)}_y[\top] =2\big)
\ \wedge\
\forall x \big(U(x) \ \to \ \#^{R(y,x)}_y[\top] \leq 1\big)
\end{align*}
Intuitively, $\Phi$ states that 
every element must belong to the unary predicate $U$
and that every element has exactly two outgoing $R$-edges
and at most one incoming $R$-edge.
It is routine to verify that an infinite binary tree 
(whose vertices are all in $U$)
satisfies $\Phi$ and that every model of $\Phi$
must be infinite. 
\end{rem}

\section{The satisfiability of \texorpdfstring{$(\gptwo)$}{GP²} with arithmetic over \texorpdfstring{$\fNinf$}{𝔑∞}}
\label{sec:main}

We introduce some terminology in Section~\ref{subsec:terminology}.
Then, we show how to represent $\gptwo$ formulas 
as graphs in Section~\ref{subsec:graph}.
The algorithm is presented in Section~\ref{subsec:algo} and 
the analysis is in Section~\ref{subsec:complexity-analysis}.
In Section~\ref{subsec:comparison} 
we make a brief comparison between our algorithm and 
the approach by Bednarczyk and Fiuk~\cite{BF22}.
Throughout this section we fix a sentence $\Psi$ 
in the normal form~(\ref{eq:normal-form})
over the signature $\Sigma$.

\subsection{Terminology}
\label{subsec:terminology}

%

A \emph{unary type} (over $\Sigma$) is 
a maximally consistent set of atomic and negated atomic formulas 
using only variable $x$, including atoms such as 
$r(x,x)$ and their negations $\neg r(x,x)$.
Similarly, a \emph{binary type} is a 
maximally consistent set of binary atoms and 
negations of atoms containing $x\neq y$,
where each atom or its negation uses two variables $x$ and $y$.
The binary type that contains only the negations of 
atomic predicates from $\Sigma$ is called the \emph{null} type,
denoted by $\etanull$. 
All the other binary types are called \emph{non-null} types.

Note that we require that each atom and the negation of an atom 
in a binary type explicitly mentions both $x$ and $y$ and $x\neq y$.
This is a little different from the standard definitions, 
such as the ones defined by Gradel, et. al.~\cite{gkv} and 
Pratt-Hartmann~\cite{c2-Pratt-Hartmann05},
where a binary type may contain unary atoms or 
negations of unary atoms involving only $x$ or $y$.
The purpose of such deviation is to make the disjointness between 
the set of unary types and the set of binary types, 
which is only for technical convenience.

Note also that both unary and binary types can be identified with 
the quantifier-free formula
formed as the conjunction of its constituent formulas.
We will use the symbols $\pi$ and $\eta$ (possibly indexed) 
to denote unary and binary types, respectively.
When viewed as formulas, we write $\pi(x)$ and $\eta(x,y)$, respectively.
We write $\pi(y)$ to denote the formula $\pi(x)$ with 
$x$ being substituted with $y$.
We denote by $\overline{\eta}(x,y)$ the ``reverse'' of $\eta(x,y)$,
i.e., the binary type obtained by swapping the variables $x$ and $y$.
Let $\Pi$ denote the set of all unary types over $\Sigma$
and let $\cK$ be the set of all {\em non-null} binary types 
over $\Sigma$.

For a structure $\cA$ (over $\Sigma$), 
the \emph{type of an element} $a \in A$ is the unique unary type $\pi$ that $a$ satisfies in $\cA$.
Similarly, the {\em type of a pair} $(a,b)\in A\times A$, where $a\neq b$, 
is the unique binary type that $(a,b)$ satisfies in $\cA$.
The {\em configuration} of a pair $(a,b)$ is the tuple $(\pi,\eta,\pi')$
where $\pi$ and $\pi'$ are the types of $a$ and $b$, respectively,
and $\eta$ is the type of $(a,b)$.
It is a non-null configuration when $\eta$ is a non-null type.
In this case we will also say that $b$ is a $(\eta,\pi')$-neighbour of $a$ in the structure $\cA$.
When $(\eta,\pi')$ is clear from the context,
we will simply say that $b$ is a {\em neighbour} of $a$.
The set of all neighbours of $a$ is called 
the {\em neighbourhood} of $a$.

We say that a unary type $\pi$ is {\em realised} in $\cA$,
if there is an element whose type is $\pi$.
Likewise, a configuration $(\pi,\eta,\pi')$ is {\em realised} in $\cA$,
if there is a pair $(a,b)$ whose configuration is $(\pi,\eta,\pi')$.

\paragraph*{The graph representation.}
In this paper the term graph means \emph{finite edge-labelled directed} graph $G=(V,E)$,
where $V\subseteq \Pi$ and $E \subseteq V\times \cK\times V$.
We can think of an edge $(\pi,\eta,\pi')$
as a potential configuration for a pair in a structure.
We allow multiple edges between two vertices provided
that they have different labels.
Similarly, we also allow multiple self-loops on a vertex provided
that they have different labels.

Let $\cN_G(\pi)$ denote the set $\{(\eta,\pi') | (\pi,\eta,\pi')\in E\}$,
i.e., the set of all the edges going out from $\pi$.
A graph $H=(V',E')$ is a \emph{subgraph} of $G=(V,E)$
if $V'\subseteq V$ and $E'\subseteq E \cap (V'\times \cK\times V')$.
If $V'\neq \emptyset$, we call $H$ a \emph{non-empty subgraph} of $G$.

Definition~\ref{def:conform} below provides the link
between structures and graphs.

\begin{defi}
\label{def:conform}
A structure $\cA$ \emph{conforms} to a graph $G$,
if all of the following conditions hold.
\begin{itemize}
\item
If a type $\pi$ is realised in $\cA$, then $\pi$ is a vertex in $G$.
\item
If a configuration $(\pi,\eta,\pi')$ is realised in $\cA$, 
where $\eta$ is a non-null type,
then $(\pi,\eta,\pi')$ is an edge in $G$.
\end{itemize}
\end{defi}

Next, we show how the sentence $\Psi$ can be represented as a graph.
Recall that $\Psi$ is a $\gptwo$ sentence in the normal form~(\ref{eq:normal-form}).
We need the following two definitions.

\begin{defi}
\label{def:compatible-vertex}
A unary type $\pi$ is \emph{compatible with $\Psi$}, if $\pi(x) \models \gamma(x)$.
\end{defi}

\begin{defi}
\label{def:compatible-edge}
A configuration $(\pi,\eta,\pi')$ is \emph{compatible with $\Psi$}, if 
both $\pi$ and $\pi'$ are compatible with $\Psi$
and for each $1\leq i\leq k$:
\[
\pi(x)\wedge \eta(x,y)\wedge \pi'(y) \models \alpha_i(x,y)
\quad\text{and}\quad
\pi(y)\wedge \overline{\eta}(x,y)\wedge \pi'(x)  \models \alpha_i(x,y).
\] 
\end{defi}

Recall that for each unary type $\pi$,
for each predicate $U(x)$,
exactly one of $U(x)$ or $\neg U(x)$ belongs to $\pi$.
Thus, to determine whether $\pi(x)\models \gamma(x)$,
it suffices to assign each atom $U(x)$ in $\gamma(x)$ 
according to the type $\pi$,
i.e., we assign an atom $U(x)$ with true, if $U(x)$ belongs to $\pi$
and false, otherwise,
and evaluate the boolean value of $\gamma(x)$ according to 
the standard semantics of $\wedge$, $\vee$ and $\neg$.
Therefore, determining whether $\pi$ is compatible with $\Psi$ 
can be done in polynomial time in the size of $\pi$ and 
the length of $\Psi$.
Similarly, since a configuration $(\pi,\eta,\pi')$ 
determines the truth value of each atom in each $\alpha_i(x,y)$,
determining whether it is compatible with a formula $\Psi$
can be done in polynomial time in the size of $\pi,\eta,\pi'$ 
and the length of $\Psi$.
Thus, listing all compatible unary types and configurations 
takes exponential time in the~length~of~$\Psi$.

The sentence $\Psi$ defines a directed graph $G_{\Psi}$,
where the vertices are the unary types that are compatible with $\Psi$
and the edges are $(\pi,\eta,\pi')$, for every $(\pi,\eta,\pi')$ compatible with~$\Psi$.
Note that the graph $G_{\Psi}$ is ``symmetric'' in the sense that
$(\pi,\eta,\pi')$ is an edge if and only if $(\pi',\overline{\eta},\pi)$ is an edge.

Intuitively, the vertices in the graph $G_{\Psi}$ 
are the unary types  that do not violate
$\forall x \ \gamma(x)$ and 
the edges are the binary types that do not violate
the conjunct
$\bigwedge_{i=1}^k \forall x \forall y \ \alpha_i(x,y)$.
Note that if a structure $\cA$ satisfies $\Psi$,
then it is necessary that $\cA$ conforms to the graph $G_{\Psi}$.
In the next section, we will show how to analyse the graph $G_{\Psi}$
to infer whether the conjunct $\bigwedge_{i=1}^{\ell} \forall x\ q_i(x) \to \cP_i(x)$
can also be satisfied.


\subsection{The characterisation of the satisfiability of \texorpdfstring{$\Psi$}{Ψ}}
\label{subsec:graph}

Recall that $\Psi$ is a sentence in normal form~(\ref{eq:normal-form}).
For each $1\leq i \leq \ell$, let the LP quantifier $\cP_i(x)$ be:
\begin{align*}
\cP_i(x) \quad & := \quad
\sum_{j=1}^{t_i} \lambda_{i,j}\cdot \#_y^{R_{i,j}(x,y)}[x\neq y]\ +\ \sum_{j=1}^{t_i'}\lambda_{i,j}'\cdot \#_y^{R_{i,j}'(x,y)}[x= y] 
\quad \circledast_i \quad
\delta_i.
\end{align*}
For a graph $G$, a vertex $\pi$ in $G$
and $1\leq i \leq \ell$,
we define the linear constraint $\cQ_i^{G,\pi}$:
\begin{align*}
\cQ_i^{G,\pi} \quad & := \quad
\sum_{j=1}^{t_i}\lambda_{i,j}\cdot 
\Bigg(\sum_{  \substack{(\eta',\pi')\in \cN_G(\pi)\\ \text{and} \ R_{i,j}(x,y)\in \eta'}}\
z_{\eta',\pi'}\Bigg) 
\ + \
\sum_{j=1}^{t_i'}\ \lambda_{i,j}'\cdot \chi_{i,j}
\quad \circledast_i \quad \delta_i,
\end{align*}
where $\chi_{i,j}$ is $1$, if $R_{i,j}'(x,x)$ is in $\pi$ and $0$, otherwise.

The variables in $\cQ_i^{G,\pi}$ are $z_{\eta',\pi'}$, for every $(\eta',\pi')\in \cK\times\Pi$.
Intuitively, each $z_{\eta',\pi'}$ represents 
the number of $(\eta',\pi')$-neighbours of an element 
with type $\pi$.
Since every element and every pair of elements has a unique unary
and binary type,
we can partition the neighbourhood of each element according to the unary and binary types.
This is the reason each $\#_{y}^{R_{i,j}(x,y)}[x\neq y]$ is replaced with the sum:
\[
\sum_{  \substack{(\eta',\pi')\ \in\ \cN_G(\pi)\\ \text{and} \ \ R_{i,j}(x,y)\ \in\ \eta'}}\ z_{\eta',\pi'}.
\]
The coefficient $\chi_{i,j}$ indicates whether an element with type $\pi$
has a $R_{i,j}'$-loop to itself.
We formalise this intuition in Lemma~\ref{lem:single-constraint}.

\begin{lem}
\label{lem:single-constraint}
Let $G$ be a graph and $\cA$ be a structure that conforms to $G$.
Let $1 \leq i \leq \ell$.
Suppose there is an element $a$ in $\cA$ whose type is $\pi$
and $\cA,x/a\models \cP_i(x)$.
Then, $\cQ_i^{G,\pi}$ admits a solution in $\fNinf$.
\end{lem}
\begin{proof}
Let $G$, $\cA$, $a$ and $\pi$ be as in the hypothesis of the lemma.
For every $(\eta',\pi')\in \cN_G(\pi)$, 
let $\cD_{\eta',\pi'}$ be the set of $(\eta',\pi')$-neighbours of $a$
in $\cA$.

Note that for every neighbour $b$ of $a$,
there is exactly one $(\eta',\pi')\in \cN_G(\pi)$ 
where $b\in \cD_{\eta',\pi'}$.
In other words,
the sets $\cD_{\eta',\pi'}$'s partition the neighbourhood of $a$.
Since $\cA,x/a\models \cP_i(x)$, by the semantics of the LP quantifier,
the following holds.
\[
\sum_{j=1}^{t_i}\ \lambda_{i,j}\cdot |R_{i,j}(x,y)\wedge x\neq y|^{x/a}_{\cA} 
\ + \
\sum_{j=1}^{t_i'}\ \lambda_{i,j}'\cdot |R_{i,j}'(x,y)\wedge x=y|^{x/a}_{\cA}
\quad \circledast_{i} \quad \delta_i.
\]
Observe also that:
\[
|R_{i,j}(x,y)\wedge x\neq y|^{x/a}_{\cA} \quad = \quad \sum_{  \substack{(\eta',\pi')\in \cN_G(\pi)\\ \text{and} \ R_{i,j}(x,y)\in \eta'}}\
|\cD_{\eta',\pi'}|.
\]
Furthermore, each $|R_{i,j}'(x,y)\wedge x=y|^{x/a}_{\cA}$ is $1$, if $R_{i,j}'(x,x)$ is in $\pi$;
and $0$, otherwise.
Hence, $|R_{i,j}'(x,y)\wedge x=y|^{x/a}_{\cA}$ is precisely the definition of $\chi_{i,j}$.
Thus, the assignment $z_{\eta',\pi'}\mapsto |\cD_{\eta',\pi'}|$ 
for each $(\eta',\pi')\in \cN_G(\pi)$ 
is a solution to the linear constraint $\cQ_i^{G,\pi}$.
\end{proof}

Next, we define a system of linear constraints that captures whether a certain configuration can be realised.
\begin{defi}
\label{def:edge-system}
For an edge $(\pi_1,\eta,\pi_2)$ in a graph $G$,
let $\cZ_{\pi_1,\eta,\pi_2}^{G}$ be the following system of linear constraints:
\[
z_{\eta,\pi_2}  \geq  1 \ \wedge 
\bigwedge_{i \ \text{s.t.}\ q_i(x) \in \pi_1} \cQ_i^{G,\pi_1}
\]
\end{defi}

Note that $\cZ_{\pi_1,\eta,\pi_2}^{G}$ contains only the linear constraint $\cQ_i^{G,\pi_1}$
when the unary predicate $q_i(x)$ belongs to $\pi_1$.
The intuitive meaning of $\cZ_{\pi_1,\eta,\pi_2}^{G}$ is as follows.
If it does not admit a solution in $\fNinf$,
then the configuration $(\pi_1,\eta,\pi_2)$
is not realised in any model of $\Psi$.
This is because either $z_{\eta,\pi_2}$ must be zero,
or $\cQ_i^{G,\pi_1}$ is violated for some $i$ where $q_i(x)\in \pi_1$.
Its formalisation is stated as Lemma~\ref{lem:system}.

\begin{lem}
\label{lem:system}
Let $G$ be a graph and $\cA$ be a structure that conforms to $G$.
Suppose there is a pair $(a,b)$ in $\cA$ whose configuration is $(\pi_1,\eta,\pi_2)$
and that $\cA,x/a\models \bigwedge_{i=1}^{\ell} (q_i(x) \to \cP_i(x))$.
Then, the system $\cZ_{\pi_1,\eta,\pi_2}^{G}$ admits a solution in $\fNinf$.
\end{lem}
\begin{proof}
Let $G$, $\cA$, $a$, $b$ and $(\pi_1,\eta,\pi_2)$ be as in the hypothesis.
Using the same notation in Lemma~\ref{lem:single-constraint},
let $\cD_{\eta',\pi'}$ denote 
the set of $(\eta',\pi')$-neighbours of $a$,
for every $(\eta',\pi') \in \cK\times \Pi$.
Since the configuration of $(a,b)$ is $(\pi_1,\eta,\pi_2)$,
we have $\cD_{\eta,\pi_2}\neq \emptyset$ and hence,
$|\cD_{\eta,\pi_2}|\geq 1$.
Since $\cA,x/a\models \bigwedge_{i=1}^{\ell} q_i(x)\to\cP_i(x)$,
it is immediate that
the assignment $z_{\eta',\pi'}\mapsto |\cD_{\eta',\pi'}|$ 
is a solution to the system $\cZ_{\pi_1,\eta,\pi_2}^G$.
\end{proof}

To infer the satisfiability of $\Psi$ from the graph $G_{\Psi}$,
we will need some more terminology.

\begin{defi}
\label{def:good-edge}
An edge $(\pi_1,\eta,\pi_2)$ is a {\em bad} edge in a graph $G$ (w.r.t. the sentence $\Psi$),
if the system $\cZ_{\pi_1,\eta,\pi_2}^{G}$ does not admit a solution in $\fNinf$.
\end{defi}

Note that by Lemma~\ref{lem:system},
if $(\pi_1,\eta,\pi_2)$ is a bad edge in $G$,
then there is no model $\cA$ that conforms to $G$ such that
the configuration $(\pi_1,\eta,\pi_2)$ is realized in $\cA$
and that $\cA \models \forall x \ (q_i(x)\to \cP_i(x))$ for every $1\leq i\leq \ell$.

Next, we define its analogue for the vertices in $G$.

\begin{defi}
\label{def:good-vertex}
A vertex $\pi$ is a {\em bad} vertex in $G$ (w.r.t. the sentence $\Psi$),
if all of the following conditions hold.
\begin{itemize}
\item
It does not have any outgoing edge in $G$.
\item
There is $1\leq i \leq \ell$ such that $\pi$ contains $q_i(x)$,
but the system $\cQ_i^{G,\pi}$ does not admit the zero solution,
i.e., the solution where all the variables are assigned with zero.
\end{itemize}
\end{defi}

The intended meaning of a bad vertex $\pi$ is that
there cannot be a model $\cA$ of $\Psi$ in which
$\pi$ is realised.
We are now ready for the final terminology 
which will be crucial for deciding the satisfiability of $\Psi$.

\begin{defi}
\label{def:good-subgraph}
Let $G$ be a graph and $H$ be a non-empty subgraph of $G$.
We say that $H$ is a {\em good} subgraph of $G$, 
if all of the following conditions hold.
\begin{itemize}
\item
There is no bad vertex and no bad edge in $H$ (w.r.t. the sentence $\Psi$).
\item
It is symmetric in the sense that $(\pi,\eta,\pi')$ is an edge in $H$ 
if and only if 
$(\pi',\overline{\eta},\pi)$ is an edge in $H$.
\end{itemize}
\end{defi}

Theorem~\ref{theo:sat-good-subgraph} states that
the satisfiability of $\Psi$ is equivalent to
the existence of a good subgraph in $G_{\Psi}$.

\begin{thm}
\label{theo:sat-good-subgraph}
Let $\Psi$ be a $\gptwo$ sentence in normal form~(\ref{eq:normal-form}).
Then, $\Psi$ is satisfiable if and only if 
there is a good subgraph in $G_{\Psi}$.
\end{thm}
\begin{proof}  {\bf (only if)} Let $\cA\models \Psi$.
Let $H$ be the graph where
the vertices are the unary types realised in $\cA$ and 
the edges are the configurations realised in $\cA$.
Obviously, $H$ is symmetric and a non-empty subgraph of $G_{\Psi}$
such that $\cA$ conforms to $H$.
It remains to show that there is no bad edge and no bad vertex in $H$.

Let $(\pi_1,\eta,\pi_2)$ be an edge in $H$.
By definition, there is a pair $(a,b)$ in $\cA$ whose configuration is 
$(\pi_1,\eta,\pi_2)$.
Since $\cA\models \Psi$, we have 
$\cA,x/a  \models  \bigwedge_{i=1}^{\ell} q_i(x) \to \cP_i(x)$.
By Lemma~\ref{lem:system}, the system $\cZ_{\pi_1,\eta,\pi_2}^{H}$ 
admits a solution in $\fNinf$,
and hence, by definition, $(\pi_1,\eta,\pi_2)$ is not a bad edge in $H$.

Next, we show that $H$ does not contain any bad vertex.
Assume to the contrary that there is a bad vertex $\pi$ in $H$.
By definition, $\pi$ does not have any outgoing edge in $H$.
By the construction of $H$,
there is an element $a$ in $\cA$ whose type is $\pi$
and for every relation $R_{i,j}(x,y)$, 
we have $
\big|R_{i,j}(x,y)\wedge x\neq y\big|^{x/a}_{\cA} = 0$.
Since $\cA\models \Psi$,
we have $\cA,x/a\models q_i(x)\to \cP_i(x)$ 
for every $1\leq i \leq \ell$.
In particular, for every $1\leq i \leq \ell$ such that 
$\pi$ contains $q_i(x)$, 
the system $\cQ_i^{H,\pi}$ admits the zero solution.
This contradicts the assumption that $\pi$ is a bad vertex.

{\bf (if)} Let $H=(V,E)$ be a good subgraph of $G_{\Psi}$.
We will show how to construct a model $\cA\models \Psi$ 
that conforms to $H$.
Let $\Psi$ be a $\gptwo$ sentence as in Eq.~(\ref{eq:normal-form}).
We consider two cases.

\underline{\em Case 1}:
There is a vertex $\pi\in V$ that does not have any outgoing edge in $H$.

Let $\cA$ be the structure that contains only one element
whose unary type is $\pi$.
Since $H$ is a subgraph of $G_{\Psi}$,
the unary type $\pi$ is compatible with $\Psi$,
i.e., $\pi(x)\models \gamma(x)$.
Hence, $\cA$ satisfies the part $\forall x \ \gamma(x)$.
Since $\cA$ contains only one element,
it trivially satisfies the part
$\bigwedge_{i=1}^k \forall x \forall y\ \alpha_i(x,y)$.

We now show that $\cA$ satisfies the third part
$\bigwedge_{i=1}^\ell \forall x\  q_i(x)  \to  \cP_i(x)$.
Note that $\pi$ is not a bad vertex since 
$H$ is a good subgraph of $G_{\Psi}$.
Since $\pi$ does not have any outgoing edge,
by definition, for every $1\leq i \leq \ell$,
where $\pi$ contains $q_i(x)$,
the system $\cQ_i^{G,\pi}$ admits the zero solution.
Since there is only one element in the structure $\cA$,
it has no neighbours and hence,
the structure $\cA$ satisfies the part
$\bigwedge_{1\leq i \leq \ell}
\forall x\  q_i(x)  \to  \cP_i(x)$.

\underline{\em Case 2}: 
Every vertex in $V$ has at least one outgoing edge in $H$.

We will build a tree-like model structure $\cA$ that satisfies $\Psi$.
In the construction of the tree,
the term \emph{$(\eta',\pi')$-children} of a node $a$
means the children of $a$ with unary type $\pi'$
and binary type of $(a,b)$ is $\eta'$.

The construction of $\cA$ is as follows.
We pick a vertex $\pi\in V$ and start with a single element $a$
and sets its unary type as $\pi$.
We then pick an arbitrary outgoing edge $(\pi,\eta,\pi_0)$ in~$H$.
Since $H$ is a good subgraph of $G_{\Psi}$,
the edge $(\pi,\eta,\pi_0)$ is a good edge.
So, by definition,
the system $\cZ_{\pi,\eta,\pi_0}^H$ admits a solution in $\fNinf$.
Let $z_{\eta',\pi'}\mapsto M_{\eta',\pi'}$ 
for every $\eta',\pi'$ be the solution.
We add ``fresh'' elements $b_1,b_2,\ldots$ as 
the children of $a$ such that for every $\eta',\pi'$,
the number of $(\eta',\pi')$-neighbours of $a$
is precisely $M_{\eta',\pi'}$.

We continue with the same process for the elements $b_1,b_2,\ldots$.
For each $j=1,2,\ldots$, let $\pi_j$ be the type of $b_j$ and
$\eta_{j}$ be the type of $(b_j,a)$.
Consider the system $\cZ_{\pi_j,\eta_j,\pi}^H$ 
which admits a solution in $\fNinf$.
Let $z_{\eta'',\pi''}\mapsto M_{\eta'',\pi''}$ 
for every $\eta'',\pi''$ be the solution.
We add ``fresh'' elements $c_{j,1},c_{j,2},\ldots$ as 
the children of $b_j$ such that for every $\eta'',\pi''$:
\begin{itemize}
\item
If $(\eta'',\pi'')=(\eta_j,\pi_j)$, 
the $(\eta'',\pi'')$-children of $b_j$ is precisely 
$M_{\eta'',\pi''}-1$.
\\
Here we need the term $-1$ since the parent node $a$
is $(\eta'',\pi'')$-neighbour of $b_j$.
\item
If $(\eta'',\pi'')\neq(\eta_j,\pi_j)$, 
the $(\eta'',\pi'')$-children of $b_j$ is precisely 
$M_{\eta'',\pi''}$. 
\end{itemize}
We repeat the same process for the elements
$c_{j,1},c_{j,2},\ldots$ ad infinitum and 
obtain an infinite tree-like model $\cA$.
We now show $\cA$ satisfies $\Psi$.
Note that $H$ is a subgraph of $G_{\Psi}$.
By the construction of $G_{\Psi}$,
every unary type and configurations are compatible with $\Psi$.
Since $\cA$ conforms to $H$,
it is immediate that $\cA$ satisfies the part $\forall x \ \gamma(x)$,
as well as the part
$\bigwedge_{i=1}^k \forall x \forall y\ \alpha_i(x,y)$.
That $\cA$ satisfies the part 
$\bigwedge_{i=1}^{\ell}\forall x \ q_i(x)  \to \cP_i(x)$
follows from the construction of the children of each element in $\cA$.
\end{proof}

\subsection{The algorithm}
\label{subsec:algo}

Theorem~\ref{theo:sat-good-subgraph} tells us that to decide if $\Psi$ is satisfiable,
it suffices to find if $G_{\Psi}$ contains a good subgraph.
It can be done as follows.
First, construct the graph~$G_{\Psi}$.
Then, repeatedly delete the bad edges and bad vertices from~$G_{\Psi}$.
It stops when there is no more bad edge or vertex to delete.
If the graph ends up not containing any vertices, 
then $\Psi$ is not satisfiable.
If the graph still contains some vertices, 
then it is a good subgraph of $G_{\Psi}$
and by Theorem~\ref{theo:sat-good-subgraph}, 
$\Psi$ is satisfiable.
Its formal presentation can be found in Algorithm~\ref{alg:gp2}.
It is worth noting that deleting a bad edge may yield a new bad edge,
hence the while-loop.

\begin{algo}\hfill
\label{alg:gp2}
\begin{algorithmic}[1]
\REQUIRE
A sentence $\Psi$ in normal form~\eqref{eq:normal-form}.
\ENSURE
Accept if and only if $\Psi$ is satisfiable.
\STATE
$G:= G_{\Psi}$.
\WHILE{$G$ has a bad edge}
\STATE
Delete the bad edge and its inverse from $G$.
\STATE
Delete all bad vertices (if there is any) from $G$.
\ENDWHILE
\STATE ACCEPT if and only if $G$ is not an empty graph.
\end{algorithmic}
\end{algo}

\subsection{The complexity analysis of Algorithm~\ref{alg:gp2}}
\label{subsec:complexity-analysis}

We start with the following lemma.
\begin{lem}
\label{lem:algo-gp2}
Algorithm~\ref{alg:gp2}
can be implemented using a quantifier-free Presburger arithmetic solver as a black box
and the number of calls to such solver is bounded by $2^{4n+8m}$,
where $n$ and $m$ is the number of unary and binary predicates
in the input sentence $\Psi$.
\end{lem}
\begin{proof}
Note that there are $2^{n+m}$ unary types and $2^{2m}$ binary types. 
Here it is useful to recall that 
atoms such as $R(x,x)$ are considered unary predicates.
Thus, the graph $G_{\Psi}$ has at most $2^{2n+4m}$ edges.
In each iteration we check whether each edge is a bad edge,
hence, there is at most $2^{2n+4m}$ number of calls to the Presburger arithmetic solver.
After each iteration there is at most one less edge,
hence, the $2^{4n+8m}$ upper bound.
\end{proof}

Next, we show that checking whether the system 
$\cZ_{\pi_1,\eta,\pi_2}^{G}$ admits a solution in $\fNinf$
can be done in non-deterministic polynomial time, and hence, in deterministic exponential time,
in the length of the input $\Psi$.

\begin{lem}
\label{lem:complexity-Z}
For every edge $(\pi_1,\eta,\pi_2)$ in $G$,
checking whether the system $\cZ_{\pi_1,\eta,\pi_2}^{G}$ admits a solution in $\fNinf$
takes non-deterministic polynomial time in the length of the input sentence.
\end{lem}
\begin{proof}
Let $\Psi$ be the input sentence in the normal form~(\ref{eq:normal-form}),
where $n$ and $m$ are the number of unary and binary relation symbols in $\Psi$.
Thus, there are $2^{n+m}$ unary types and $2^{2m}$ binary types. 
Moreover, the system $\cZ_{\pi_1,\eta,\pi_2}^{G}$ contains at most $\ell+1$ linear constraints and $2^{n+3m}$ variables.

Here we invoke Corollary~\ref{cor:caratheodory-inf} which states that
there are constants $c_1,c_2$ such that
if $\cZ_{\pi_1,\eta,\pi_2}^{G}$ admits a solution in $\fNinf$,
then the number of variables assigned with non-zero values is $c_1\ell\log_2(c_2\ell M))$
and each value is either $\infty$ or at most $c_1\ell(\ell M)^{c_2\ell}$,
where $M$ is the maximal non-$\infty$ constant in the system $\cZ_{\pi_1,\eta,\pi_2}^{G}$.
Note that the value $c_1\ell(\ell M)^{c_2\ell}$ takes only polynomially many bits in the length of $\Psi$.

We can thus design a non-deterministic polynomial time algorithm for checking
the satisfiability of $\cZ_{\pi_1,\eta,\pi_2}^{G}$:
Guess the set of variables that are supposed to take non-zero values
and their values 
and ACCEPT if and only if it is indeed a solution for $\cZ_{\pi_1,\eta,\pi_2}^{G}$.
\end{proof}

Since constructing the graph $G_{\Psi}$ takes exponential time,
combining it with Lemmas~\ref{lem:algo-gp2} and~\ref{lem:complexity-Z}
implies the exponential time upper bound of Algorithm~\ref{alg:gp2}.

\begin{thm}
\label{theo:algo-run-time}
Algorithm~\ref{alg:gp2} runs in exponential time in the length of the input sentence.
\end{thm}

We remark that the exponential time upper bound of Algorithm~\ref{alg:gp2} 
also holds when the semantics of LP quantifiers is defined on the structure $\fN$.
The proof is exactly the same.
The only difference is we 
invoke Corollary~\ref{cor:caratheodory} instead of
Corollary~\ref{cor:caratheodory-inf}.

\subsection{Comparison with the approach by Bednarczyk and Fiuk~\texorpdfstring{\cite{BF22}}{[\cite{BF22}]}}
\label{subsec:comparison}

An alternating polynomial space algorithm was proposed Bednarczyk and Fiuk~\cite{BF22}
for deciding the satisfiability of $\gptwo$ sentences.\footnote{It is well known that
the class of languages decidable by alternating polynomial space Turing machines
is equivalent to the class of languages decidable by 
deterministic exponential time Turing machines~\cite{alternation}.}
Intuitively, it tries to build a tree-like model
that satisfies the input sentence (if satisfiable).
It starts by guessing the unary type of the root node
and all the unary types of its children as well as
the binary types of the edges connecting them.
Using Theorem~\ref{theo:caratheodory},
the number of different \emph{unary types} of the children 
is polynomially bounded (in the length of the input sentence).\footnote{Note
that the bound only holds for the number of unary types of the children,
not the number of children itself which can be exponential.}
To verify that the LP quantifiers are satisfied,
it guesses a solution that respects the guessed binary types of
the edges connecting the root nodes and the children.

In comparison, our algorithm is a straightforward 
deterministic algorithm
that can be implemented easily
using available Presburger solvers as a black box.
Nevertheless our algorithm is worst-case optimal,
as even in the best case scenario it still requires exponential time 
to build the graph representation,
in contrast to the tableaux based approaches
whose performance can be efficient depending on the instances.

%

\section{Concluding remarks}
\label{sec:conclusion}

In this paper we consider the extension of $\gftwo$
with the local Presburger quantifiers which can express rich Presburger constraints
while maintaining a deterministic exponential time upper bound.
It captures various natural DLs with counting up to
$\mathcal{ALCIH}b^{\textsf{self}}$ and $\mathcal{ALCIHQ}$ without constant symbols.
The proof is via a novel, yet simple algorithm,
which is reminiscent of the type-elimination approach.

We note that the proof of Theorem~\ref{theo:sat-good-subgraph} relies on the infinity of the model.
We leave a similar characterization for the finite models for future work.
It is worth noting that the finite satisfiability of $\gptwo$ 
has been shown to be decidable in 3-$\nexp$~\cite{BOPT21},
though the precise complexity is still not known.

\section*{Acknowledgment}
\noindent 
We thank Bartosz Bednarczyk for his comments on the preliminary drafts of this work. 
We would like to thank the anonymous reviewers for their detailed 
and constructive feedback which greatly improve the presentation 
of this paper.
This work was done when both authors were in National Taiwan University.
We acknowledge generous financial support of 
Taiwan Ministry of Science and Technology under grant no.~109-2221-E-002-143-MY3.

\bibliographystyle{alphaurl}
\bibliography{aa-bib-c2-algo}

\end{document}